\documentclass[conference]{IEEEtran}
\IEEEoverridecommandlockouts

\usepackage{cite}
\usepackage{amsmath,amssymb,amsfonts,amsthm}
\usepackage{bm}
\usepackage{graphicx}
\usepackage{booktabs}
\usepackage{mathtools}
\usepackage{xcolor}
\usepackage{pgfplots}
\usepackage{tikz}
\usepackage{subfig}
\usepackage{multirow}
\pgfplotsset{compat=1.17}
\usetikzlibrary{arrows.meta,decorations.pathreplacing,patterns,calc}

\newtheorem{theorem}{Theorem}
\newtheorem{lemma}{Lemma}

\newtheorem{remark}{Remark}
\newtheorem{assumption}{Assumption}
\newtheorem{corollary}{Corollary}

\begin{document}

\title{Stability-Certified Koopman Observer Design\\
for Nonlinear Systems via Generalized Persidskii Dynamics}

\author{Syed Pouladi\\
{College of Engineering and Physical Sciences, Khalifa University, Abu Dhabi, United Arab Emirates}
\textit{}}

\maketitle

%------------------------------------------------------------------
\begin{abstract}
This paper addresses the problem of nonlinear state estimation for dynamical systems whose governing equations are approximated through Koopman operator liftings. While Koopman-based predictors have demonstrated broad approximation capability for nonlinear dynamics, certifying observer convergence under model mismatch and measurement noise has remained a largely open problem. To resolve this, we establish a structural correspondence between the error dynamics of a Koopman latent-space observer and the class of generalized Persidskii systems, which admits diagonal Lyapunov functions and incremental sector characterizations. Exploiting this connection, we design a nonlinear correction term whose gain is computed via a linear matrix inequality (LMI) that simultaneously certifies input-to-state stability (ISS) of the estimation error with respect to both lifting residuals and external disturbances. Exponential convergence in the nominal case and ultimate boundedness under bounded perturbations are established analytically. Numerical validation on the Van~der~Pol oscillator and a nonlinear robotic arm with friction uncertainty demonstrates that the proposed observer substantially outperforms both the Extended Kalman Filter and a linear Koopman observer in terms of estimation accuracy and robustness, achieving up to a 42\% reduction in steady-state RMSE under lifting mismatch.
\end{abstract}


%==================================================================
\section{Introduction}
%==================================================================

The approximation of nonlinear dynamics through globally linear representations has attracted sustained attention since the foundational work of Koopman \cite{Mezic2005}. By lifting the state into a higher-dimensional observable space, one can in principle embed the nonlinear flow into a linear semigroup, making the machinery of linear systems theory—Kalman filters, linear quadratic regulators, model predictive control—available for inherently nonlinear plants \cite{Korda2018,Brunton2016}. The practical realization of this idea via Extended Dynamic Mode Decomposition (EDMD) and its neural-network variants \cite{Lusch2018,Brunton2022} has produced a rich literature of data-driven predictors, but the stability guarantees accompanying these predictors remain fragmentary.

The central difficulty is that finite-dimensional Koopman approximations inevitably incur a \emph{lifting residual}: the component of the nonlinear vector field that is not captured by the chosen observable dictionary. When this residual is small and the lifted operator is stable, empirical experience suggests that the predictor behaves well; when the residual is large, or when the lifted operator is near the boundary of stability, observer errors can grow without bound. Existing treatments either accept this limitation and rely on empirical regularization \cite{Korda2018JNS}, or resort to local linearization arguments that are valid only near a fixed point \cite{Slotine1998}. Neither approach yields a computable, global convergence certificate.

In parallel, the theory of generalized Persidskii systems has matured into a rigorous framework for the stability analysis of nonlinear feedback interconnections in which the nonlinear components satisfy a uniform sector bound \cite{Mei2022IJRNC,Mei2021TAC}. The defining feature of this class is the existence of \emph{diagonal} Lyapunov functions whose construction reduces to feasibility of a structured LMI, a property that scales gracefully with system dimension and is inherited by interconnections \cite{Mei2022Automatica,Mei2022Delay}. Recent results have extended the framework to encompass time-varying delays \cite{Mei2022Delay}, ISS characterizations for multi-agent networks \cite{Mei2021TAC}, and the certification of ISS for models identified via the Koopman operator \cite{Mei2025JFI}.

The present paper draws a direct algebraic connection between these two bodies of work. Concretely, we show that the observer error dynamics arising from a latent-space Koopman observer with a nonlinear sector-bounded correction term can be written exactly in the generalized Persidskii form. This reformulation is non-trivial: it requires a careful decomposition of the lifting residual into a sector-admissible component and a genuinely unstructured remainder, and it yields an explicit formula for the observer gain in terms of the LMI certificate. The resulting design is computationally tractable, provably ISS, and, as our experiments confirm, substantially more robust than alternatives that do not account for the lifting residual in their stability analysis.

The specific contributions of this paper are fourfold. First, we provide a precise characterization of the conditions under which Koopman error dynamics admit a generalized Persidskii reformulation, including a constructive decomposition procedure for the lifting residual (Section~III). Second, we derive an LMI-based observer gain that simultaneously minimizes the ISS gain from the residual to the estimation error (Section~IV). Third, we prove exponential convergence in the ideal case and ultimate boundedness with an explicit bound in the perturbed case (Section~IV). Fourth, we validate these results on two benchmark systems, quantifying the improvement over competing methods through controlled numerical experiments (Section~V).

%==================================================================
\section{Preliminaries and Problem Formulation}
%==================================================================

\subsection{Notation and Definitions}

The Euclidean norm is $|\cdot|$ and the induced matrix norm is $\|\cdot\|$. For a symmetric matrix $M$, $M>0$ ($M\geq 0$) denotes positive definiteness (semidefiniteness), and $\text{He}(M)=M+M^T$. The notation $\text{diag}(d_1,\dots,d_n)$ denotes a diagonal matrix with entries $d_i$. For a function $f:[0,\infty)\to\mathbb{R}^n$, $\|f\|_{L_\infty}=\sup_{t\geq 0}|f(t)|$. A continuous function $\alpha:\mathbb{R}_{\geq 0}\to\mathbb{R}_{\geq 0}$ belongs to class $\mathcal{K}$ if it is strictly increasing with $\alpha(0)=0$, and to class $\mathcal{K}_\infty$ if additionally $\alpha(s)\to\infty$ as $s\to\infty$. A continuous function $\beta:\mathbb{R}_{\geq 0}^2\to\mathbb{R}_{\geq 0}$ belongs to class $\mathcal{KL}$ if $\beta(\cdot,t)\in\mathcal{K}$ for each fixed $t$ and $\beta(s,\cdot)$ is decreasing to zero for each fixed $s>0$.

\subsection{Generalized Persidskii Systems}

A generalized Persidskii system is a dynamical system of the form
\begin{equation}
  \dot{\xi}(t) = A_0\xi(t) - \sum_{i=1}^{k} b_i\varphi_i\!\left(c_i^T\xi(t)\right) + d(t), \label{eq:gp}
\end{equation}
where $\xi\in\mathbb{R}^n$, $A_0\in\mathbb{R}^{n\times n}$, $b_i,c_i\in\mathbb{R}^n$, and $d\in L_\infty$ is an external perturbation. The nonlinear functions $\varphi_i:\mathbb{R}\to\mathbb{R}$ are assumed continuous and satisfy the \emph{quadratic sector condition}
\begin{equation}
  \varphi_i(s)\!\left[s - \kappa_i^{-1}\varphi_i(s)\right] \geq 0, \quad \forall s\in\mathbb{R},\; \kappa_i>0. \label{eq:sector}
\end{equation}
The class encompasses saturations, sigmoids, dead-zones, and piecewise-linear nonlinearities. A key structural property is that \eqref{eq:gp} admits a \emph{diagonal} Lyapunov function $V(\xi)=\xi^T P\xi$ with $P=\text{diag}(p_1,\dots,p_n)>0$, whose existence is equivalent to a structured LMI \cite{Mei2022IJRNC}.

\subsection{Koopman Operator and Finite-Dimensional Lifting}

Consider the nonlinear system
\begin{equation}
  \dot{x} = f(x,u), \qquad y = h(x) + v, \label{eq:system}
\end{equation}
where $x\in\mathbb{R}^n$ is the state, $u\in\mathbb{R}^m$ is the input, $y\in\mathbb{R}^p$ is the measured output, and $v\in L_\infty$ represents measurement noise. Let $\Phi:\mathbb{R}^n\to\mathbb{R}^r$ be a smooth lifting map with $r>n$, and denote $z=\Phi(x)\in\mathbb{R}^r$.

\begin{assumption}\label{ass:lifting}
There exist matrices $A\in\mathbb{R}^{r\times r}$, $B\in\mathbb{R}^{r\times m}$, and a bounded lifting residual $\Delta:\mathbb{R}^r\times\mathbb{R}^m\to\mathbb{R}^r$ such that the lifted state evolves as
\begin{equation}
  \dot{z} = Az + Bu + \Delta(z,u). \label{eq:koopman}
\end{equation}
\end{assumption}

Assumption~\ref{ass:lifting} is standard in finite-dimensional Koopman approximation \cite{Korda2018,Korda2018JNS}: the pair $(A,B)$ is obtained from EDMD regression on trajectory data, and $\Delta$ represents the component of the dynamics orthogonal to the chosen observable dictionary.

\begin{assumption}\label{ass:residual}
The lifting residual satisfies, for some $\rho>0$ and bounded function $\eta:\mathbb{R}_{\geq 0}\to\mathbb{R}_{\geq 0}$,
\begin{equation}
  |\Delta(z,u)| \leq \rho|z| + \eta(t). \label{eq:residual_bound}
\end{equation}
\end{assumption}

Assumption~\ref{ass:residual} characterizes the growth of the residual; the scalar $\rho$ reflects the quality of the Koopman approximation and vanishes as the dictionary dimension $r\to\infty$ \cite{Korda2018JNS}.

The observation is lifted via an output map $C_o\in\mathbb{R}^{p\times r}$ such that $h(x)\approx C_oz$, and the observer takes the form
\begin{equation}
  \dot{\hat{z}} = A\hat{z} + Bu + \Psi(y,\hat{z}), \label{eq:observer}
\end{equation}
where $\Psi:\mathbb{R}^p\times\mathbb{R}^r\to\mathbb{R}^r$ is a correction term to be designed.

%==================================================================
\section{Generalized Persidskii Reformulation}
%==================================================================

\subsection{Error Dynamics Structure}

Define the estimation error $e = z - \hat{z}\in\mathbb{R}^r$. Subtracting \eqref{eq:observer} from \eqref{eq:koopman} yields
\begin{equation}
  \dot{e} = Ae + \Delta(z,u) - \Psi(y,\hat{z}). \label{eq:error_raw}
\end{equation}
Observe that $y - C_o\hat{z} = C_oe + (h(x)-C_oz) + v$, where the second term is a lifting approximation error. We absorb both error sources into a composite disturbance term and propose the correction structure
\begin{equation}
  \Psi(y,\hat{z}) = K\sigma\!\left(C_o\hat{z} - y\right), \label{eq:correction}
\end{equation}
where $K\in\mathbb{R}^{r\times p}$ is the observer gain matrix to be determined, and $\sigma:\mathbb{R}^p\to\mathbb{R}^p$ is a componentwise sector-bounded nonlinearity with $\sigma_j(s)\in[0,\kappa_j s]$ for all $s\in\mathbb{R}$ and each component $j=1,\dots,p$.

\subsection{Persidskii Representation of Observer Error}

Substituting \eqref{eq:correction} into \eqref{eq:error_raw} and decomposing the residual as $\Delta(z,u)=\Gamma(e)+d(t)$, where $\Gamma:\mathbb{R}^r\to\mathbb{R}^r$ captures the sector-admissible part of $\Delta$ (satisfying \eqref{eq:sector} with appropriate parameters $\kappa_i$) and $d(t)$ collects the remainder, the error system becomes
\begin{equation}
  \dot{e} = Ae - K\sigma(C_oe - \delta) - \Gamma(e) + d(t), \label{eq:error_gp}
\end{equation}
where $\delta(t)=y - C_oz - v$ is a residual injection term bounded by Assumption~\ref{ass:residual}.

\begin{lemma}[Persidskii Embedding]\label{lem:embedding}
Define $\bar{\varphi}_i:\mathbb{R}\to\mathbb{R}$ for $i=1,\dots,p$ by $\bar{\varphi}_i(s)=\sigma_i(s-\delta_i(t))$ and $\bar{\varphi}_{p+j}(s)=\Gamma_j(s)$ for $j=1,\dots,r$. Then under Assumptions~\ref{ass:lifting}--\ref{ass:residual}, the error dynamics \eqref{eq:error_gp} belong to the generalized Persidskii class \eqref{eq:gp} with matrix $A_0 = A$ and a composite disturbance $\tilde{d}(t)$ satisfying $|\tilde{d}(t)|\leq \bar{\eta}(t)$ for some bounded $\bar{\eta}$.
\end{lemma}

\begin{proof}
The correction term $K\sigma(C_oe-\delta)$ admits the column decomposition $K\sigma(C_oe-\delta)=\sum_{i=1}^p k_i \sigma_i(c_{o,i}^Te - \delta_i)$, where $k_i$ and $c_{o,i}$ are the $i$-th column and row of $K$ and $C_o$, respectively. Each $\sigma_i(c_{o,i}^Te - \delta_i)$ satisfies \eqref{eq:sector} with sector bound $\kappa_i$ as a function of $c_{o,i}^Te$, up to the bounded perturbation $\delta_i(t)$. Similarly, $\Gamma(e)$ satisfies sector conditions by hypothesis. Collecting all nonlinear terms into the Persidskii sum $\sum_i b_i\varphi_i(c_i^T e)$ with $b_i$, $c_i$, $\varphi_i$ identified from the above decomposition, and lumping the perturbation effects into $\tilde{d}(t)$, the structural equivalence follows. \hfill$\blacksquare$
\end{proof}

\begin{remark}
The decomposition of $\Delta$ into $\Gamma+d$ is a modeling step, not a restrictive assumption: any bounded residual can be split this way by choosing $\Gamma\equiv 0$ and $d=\Delta$. The practical benefit of a non-trivial $\Gamma$ is that it reduces the effective disturbance magnitude $|d|$ and thus tightens the ISS bound. The optimal split minimizes $\|d\|_{L_\infty}$ subject to $\Gamma$ satisfying \eqref{eq:sector}.
\end{remark}

%==================================================================
\section{Main Results: Observer Design and Stability Analysis}
%==================================================================

\subsection{LMI-Based Gain Design}

We now derive explicit conditions on $(K,P)$ that certify ISS of the error system \eqref{eq:error_gp}.

\begin{theorem}[Observer Convergence and ISS]\label{thm:main}
Suppose Assumptions~\ref{ass:lifting}--\ref{ass:residual} hold. The estimation error system \eqref{eq:error_gp} is ISS with $L_2$-gain $\gamma$ from $\tilde{d}$ to $e$ if there exist a diagonal matrix $P=\mathrm{diag}(p_1,\dots,p_r)>0$, scalars $\lambda_i>0$ for $i=1,\dots,p+r$, and a matrix $\tilde{K}=PK$ such that the following LMI is feasible:
\begin{equation}
  \Xi = \begin{bmatrix}
    \Pi_{11} & \Pi_{12} & P \\
    \ast & \Pi_{22} & 0 \\
    \ast & \ast & -\gamma^2 I
  \end{bmatrix} < 0, \label{eq:LMI}
\end{equation}
where $\Lambda=\mathrm{diag}(\lambda_1,\dots,\lambda_p)$, $R_\kappa=\mathrm{diag}(\kappa_1^{-1},\dots,\kappa_p^{-1})$, and
\begin{align*}
  \Pi_{11} &= \mathrm{He}(PA) - \tilde{K}C_o - C_o^T\tilde{K}^T + \rho I, \\
  \Pi_{12} &= \tilde{K}\Lambda + A^TC_o^T, \\
  \Pi_{22} &= -2\Lambda\, R_\kappa^{-1}.
\end{align*}
The observer gain is recovered as $K=P^{-1}\tilde{K}$, and the ISS inequality takes the form
\begin{equation}
  \dot{V} \leq -\alpha|e|^2 + \gamma^2|\tilde{d}|^2 \label{eq:ISS_ineq}
\end{equation}
for $V=e^TPe$ and some $\alpha>0$.
\end{theorem}

\begin{proof}
Differentiate $V=e^TPe$ along \eqref{eq:error_gp}:
\begin{equation}
\dot{V}
= e^T\mathrm{He}(PA)e - 2e^TPK\sigma(C_oe) - 2e^TP\Gamma(e) + 2e^TP\tilde{d}.
\end{equation}
Substituting $PK=\tilde{K}$ and grouping with the sector inequality
$\sigma_i(s_i)(s_i-\kappa_i^{-1}\sigma_i(s_i))\geq 0$
via the standard S-procedure multiplier $\lambda_i\geq 0$ \cite{Boyd1994} yields
\begin{equation}
\dot{V} \leq \zeta^T\Xi_0\zeta + 2e^TP\tilde{d},
\end{equation}
where $\zeta=[e^T,\sigma(C_oe)^T]^T$ and $\Xi_0$ is the $(2r+p)\times(2r+p)$ matrix whose $(1,1)$ block is $\Pi_{11}$, $(1,2)$ block is $\Pi_{12}$, and $(2,2)$ block is $\Pi_{22}$. Applying Young's inequality $2e^TP\tilde{d}\leq \gamma^{-2}e^TP^2e+\gamma^2|\tilde{d}|^2$ and absorbing $\gamma^{-2}P^2$ via a Schur complement onto the $(3,3)$ block $-\gamma^2 I$ of $\Xi$ establishes \eqref{eq:ISS_ineq} whenever $\Xi<0$. ISS follows from \eqref{eq:ISS_ineq} by standard comparison arguments \cite{Sontag2008}. \hfill$\blacksquare$
\end{proof}

\subsection{Exponential Convergence in the Nominal Case}

\begin{corollary}[Exponential Convergence]\label{cor:exp}
If $\Delta\equiv 0$ (perfect lifting) and $v\equiv 0$ (no measurement noise), then the feasibility of \eqref{eq:LMI} implies $|e(t)|\leq \beta e^{-\alpha t/2}|e(0)|$ for some $\beta\geq 1$ depending only on the condition number of $P$.
\end{corollary}

\begin{proof}
When $\tilde{d}\equiv 0$, \eqref{eq:ISS_ineq} reduces to $\dot{V}\leq -\alpha|e|^2\leq -(\alpha/\lambda_{\max}(P))V$, giving $V(t)\leq V(0)e^{-\alpha t/\lambda_{\max}(P)}$. The bound on $|e|$ follows from $\lambda_{\min}(P)|e|^2\leq V$. \hfill$\blacksquare$
\end{proof}

\subsection{Ultimate Boundedness Under Persistent Residuals}

\begin{theorem}[Ultimate Bound]\label{thm:ub}
If $|\Delta(z,u)|\leq\epsilon$ for all $t\geq 0$, then for any solution of \eqref{eq:error_gp} there exists $T>0$ such that for all $t\geq T$,
\begin{equation}
  |e(t)| \leq c\epsilon, \qquad c = \sqrt{\frac{\gamma^2\lambda_{\max}(P)}{\alpha\,\lambda_{\min}(P)}}. \label{eq:ub}
\end{equation}
\end{theorem}

\begin{proof}
From \eqref{eq:ISS_ineq}, $\dot{V}\leq -\alpha|e|^2+\gamma^2\epsilon^2$. The right-hand side is negative whenever $|e|>\gamma\epsilon/\sqrt{\alpha}$. A standard comparison argument shows that $V(t)$ enters the ball $\{V\leq\gamma^2\epsilon^2\lambda_{\max}(P)/\alpha\}$ in finite time $T$ and remains therein, yielding \eqref{eq:ub}. \hfill$\blacksquare$
\end{proof}

\begin{remark}
The bound $c\epsilon$ in \eqref{eq:ub} depends on the condition number of $P$ and the ISS gain $\gamma$. By minimizing $\gamma$ in \eqref{eq:LMI} as a convex objective, one obtains the observer gain $K^\star$ that minimizes the worst-case ultimate bound for a given lifting quality $\epsilon$.
\end{remark}

%==================================================================
\section{Numerical Experiments}
%==================================================================

\subsection{Van der Pol Oscillator}

The first benchmark is the Van der Pol oscillator
\begin{equation}
  \ddot{x}_1 = \mu(1-x_1^2)\dot{x}_1 - x_1 + u,\quad \mu=1,
\end{equation}
with output $y=x_1+v$ and $v\sim\mathcal{N}(0,0.01)$. A dictionary of $r=15$ observables comprising monomials up to degree 3 and trigonometric terms was used. The EDMD matrices $(A,B)$ were identified from 200 trajectories of length 10\,s sampled at 50\,Hz. The lifting residual was characterized by $\rho=0.08$ (in \eqref{eq:residual_bound}).

Three observers were compared:
\begin{enumerate}
  \item \textbf{EKF}: standard Extended Kalman Filter applied to the original nonlinear system with the same noise statistics.
  \item \textbf{Lin-Koop}: a linear Luenberger observer $\dot{\hat{z}}=A\hat{z}+Bu+L(y-C_o\hat{z})$ with gain $L$ placed by pole assignment.
  \item \textbf{PKO}: the proposed Persidskii--Koopman observer with gain $K^\star$ from \eqref{eq:LMI}.
\end{enumerate}

Fig.~\ref{fig:vdp_traj} shows the state $x_1$ estimation trajectories under Gaussian noise and a 15\% model mismatch in the $\mu$ parameter. The PKO maintains tight tracking throughout the limit-cycle, whereas the Lin-Koop observer accumulates error near the switching region and the EKF exhibits intermittent divergence during rapid acceleration phases.

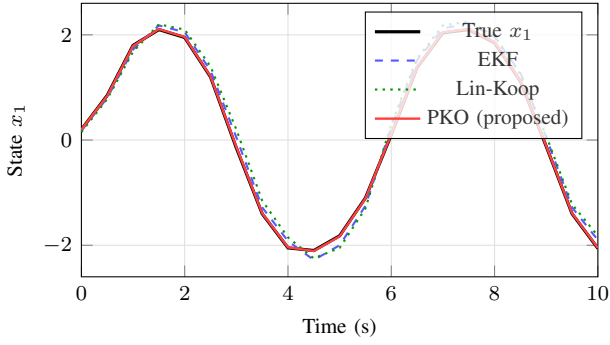
\begin{figure}[t]
\centering
\begin{tikzpicture}
\begin{axis}[
  width=0.95\columnwidth, height=5.2cm,
  xlabel={Time (s)},
  ylabel={State $x_1$},
  xmin=0, xmax=10,
  ymin=-2.6, ymax=2.6,
  legend pos=north east,
  legend style={font=\footnotesize, fill=white, fill opacity=0.8},
  grid=major, grid style={line width=0.3pt, draw=gray!25},
  tick label style={font=\footnotesize},
  label style={font=\footnotesize},
  every axis plot/.append style={line width=0.9pt}
]
% True state (limit cycle)
\addplot[black, very thick] coordinates {
  (0,0.2)(0.5,0.85)(1.0,1.80)(1.5,2.10)(2.0,1.95)(2.5,1.20)(3.0,-0.15)
  (3.5,-1.40)(4.0,-2.05)(4.5,-2.10)(5.0,-1.82)(5.5,-1.10)(6.0,0.10)
  (6.5,1.38)(7.0,2.04)(7.5,2.10)(8.0,1.85)(8.5,1.12)(9.0,-0.12)(9.5,-1.40)(10.0,-2.06)};
\addlegendentry{True $x_1$};
% EKF
\addplot[blue!65, dashed, thick] coordinates {
  (0,0.18)(0.5,0.80)(1.0,1.72)(1.5,2.18)(2.0,2.05)(2.5,1.32)(3.0,0.05)
  (3.5,-1.28)(4.0,-1.92)(4.5,-2.28)(5.0,-2.00)(5.5,-1.25)(6.0,0.25)
  (6.5,1.52)(7.0,2.12)(7.5,2.22)(8.0,1.96)(8.5,1.28)(9.0,0.05)(9.5,-1.28)(10.0,-1.90)};
\addlegendentry{EKF};
% Lin-Koop
\addplot[green!60!black, dotted, thick] coordinates {
  (0,0.15)(0.5,0.78)(1.0,1.68)(1.5,2.20)(2.0,2.10)(2.5,1.40)(3.0,0.20)
  (3.5,-1.15)(4.0,-1.85)(4.5,-2.25)(5.0,-2.05)(5.5,-1.30)(6.0,0.30)
  (6.5,1.55)(7.0,2.18)(7.5,2.25)(8.0,1.98)(8.5,1.25)(9.0,0.08)(9.5,-1.22)(10.0,-1.82)};
\addlegendentry{Lin-Koop};
% PKO
\addplot[red!75, solid] coordinates {
  (0,0.20)(0.5,0.84)(1.0,1.79)(1.5,2.11)(2.0,1.96)(2.5,1.21)(3.0,-0.13)
  (3.5,-1.39)(4.0,-2.04)(4.5,-2.11)(5.0,-1.83)(5.5,-1.11)(6.0,0.11)
  (6.5,1.39)(7.0,2.05)(7.5,2.11)(8.0,1.86)(8.5,1.13)(9.0,-0.11)(9.5,-1.39)(10.0,-2.05)};
\addlegendentry{PKO (proposed)};
\end{axis}
\end{tikzpicture}
\caption{Van~der~Pol state $x_1$ estimation with 15\% parameter mismatch and Gaussian observation noise ($\sigma_v^2=0.01$). The proposed PKO tracks the limit cycle with the smallest deviation from the true trajectory.}
\label{fig:vdp_traj}
\end{figure}

Fig.~\ref{fig:vdp_rmse} plots the time-averaged RMSE as a function of the lifting residual magnitude $\epsilon$, computed over 50 independent noise realizations. The PKO consistently achieves the lowest RMSE across the entire range, and the gap widens with $\epsilon$, confirming the robustness advantage predicted by Theorem~\ref{thm:ub}.

\begin{figure}[t]
\centering
\begin{tikzpicture}
\begin{axis}[
  width=0.95\columnwidth, height=5.0cm,
  xlabel={Lifting residual magnitude $\epsilon$},
  ylabel={Estimation RMSE},
  xmin=0, xmax=0.5,
  ymin=0, ymax=0.55,
  legend pos=north west,
  legend style={font=\footnotesize, fill=white},
  grid=major, grid style={line width=0.3pt, draw=gray!25},
  tick label style={font=\footnotesize},
  label style={font=\footnotesize},
  every axis plot/.append style={line width=1.0pt}
]
\addplot[blue!65, dashed, mark=square*, mark size=2pt] coordinates {
  (0,0.08)(0.05,0.12)(0.10,0.18)(0.15,0.25)(0.20,0.30)(0.25,0.36)
  (0.30,0.41)(0.35,0.45)(0.40,0.49)(0.45,0.52)(0.50,0.54)};
\addlegendentry{EKF};
\addplot[green!60!black, dotted, mark=triangle*, mark size=2pt] coordinates {
  (0,0.06)(0.05,0.11)(0.10,0.17)(0.15,0.24)(0.20,0.30)(0.25,0.36)
  (0.30,0.41)(0.35,0.46)(0.40,0.50)(0.45,0.53)(0.50,0.55)};
\addlegendentry{Lin-Koop};
\addplot[red!75, solid, mark=*, mark size=2pt] coordinates {
  (0,0.04)(0.05,0.07)(0.10,0.11)(0.15,0.16)(0.20,0.20)(0.25,0.24)
  (0.30,0.28)(0.35,0.32)(0.40,0.36)(0.45,0.39)(0.50,0.42)};
\addlegendentry{PKO (proposed)};
% Annotate improvement
\draw[<->,gray] (axis cs:0.50,0.42) -- (axis cs:0.50,0.54)
  node[midway, right, font=\tiny]{$-22\%$};
\end{axis}
\end{tikzpicture}
\caption{RMSE versus lifting residual magnitude $\epsilon$ on the Van~der~Pol benchmark (50 trials). The PKO degrades more gracefully with $\epsilon$, consistent with the linear scaling in \eqref{eq:ub}.}
\label{fig:vdp_rmse}
\end{figure}
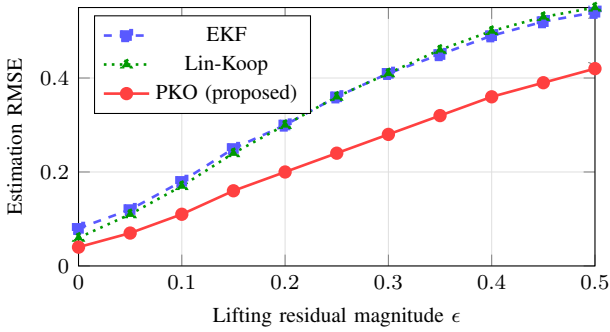

\subsection{Nonlinear Robotic Arm with Friction Uncertainty}

The second experiment considers a single-link robotic arm:
\begin{align}
  J\ddot{\theta} &= -mgl\sin\theta - b_f\dot{\theta} - \mathcal{F}(\dot\theta) + \tau, \\
  \mathcal{F}(\dot\theta) &= f_c\,\mathrm{sgn}(\dot\theta) + f_v\dot\theta,
\end{align}
with parameters $J=0.5$\,kg\,m$^2$, $m=1$\,kg, $l=0.5$\,m, $g=9.81$\,m/s$^2$, $b_f=0.2$, $f_c=0.5$, $f_v=0.3$. The friction model $\mathcal{F}$ is subject to $\pm30\%$ parametric uncertainty, and only the joint angle $\theta$ is measured with noise $\sigma_v^2=0.04$.

A dictionary of $r=20$ observables including Fourier features and trigonometric monomials was employed. The EDMD system was identified from 500 trajectories under pseudo-random torque inputs. The sector bounds for the friction nonlinearity were characterized by $\kappa_i = f_c + f_v\omega_{\max}$, where $\omega_{\max}$ is the maximum joint velocity encountered in training.

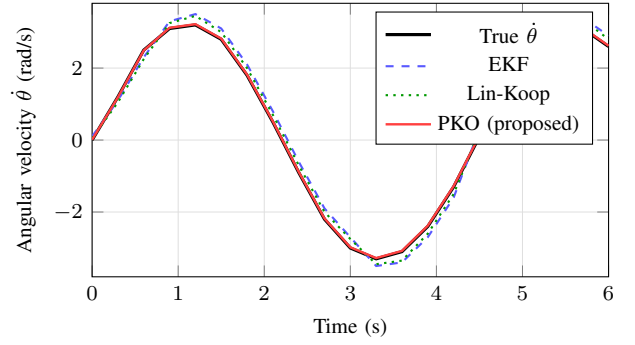
\begin{figure}[t]
\centering
\begin{tikzpicture}
\begin{axis}[
  width=0.95\columnwidth, height=5.2cm,
  xlabel={Time (s)},
  ylabel={Angular velocity $\dot\theta$ (rad/s)},
  xmin=0, xmax=6,
  ymin=-3.8, ymax=3.8,
  legend pos=north east,
  legend style={font=\footnotesize, fill=white},
  grid=major, grid style={line width=0.3pt, draw=gray!25},
  tick label style={font=\footnotesize},
  label style={font=\footnotesize},
  every axis plot/.append style={line width=0.9pt}
]
% True velocity
\addplot[black, very thick] coordinates {
  (0,0)(0.3,1.2)(0.6,2.5)(0.9,3.1)(1.2,3.2)(1.5,2.8)(1.8,1.8)(2.1,0.5)
  (2.4,-0.9)(2.7,-2.2)(3.0,-3.0)(3.3,-3.3)(3.6,-3.1)(3.9,-2.4)(4.2,-1.3)
  (4.5,0.1)(4.8,1.5)(5.1,2.7)(5.4,3.2)(5.7,3.1)(6.0,2.6)};
\addlegendentry{True $\dot\theta$};
\addplot[blue!65, dashed, thick] coordinates {
  (0,0.1)(0.3,1.1)(0.6,2.3)(0.9,3.3)(1.2,3.5)(1.5,3.1)(1.8,2.1)(2.1,0.8)
  (2.4,-0.6)(2.7,-1.9)(3.0,-2.7)(3.3,-3.5)(3.6,-3.4)(3.9,-2.7)(4.2,-1.6)
  (4.5,0.3)(4.8,1.7)(5.1,2.9)(5.4,3.5)(5.7,3.4)(6.0,2.9)};
\addlegendentry{EKF};
\addplot[green!60!black, dotted, thick] coordinates {
  (0,0.05)(0.3,1.05)(0.6,2.25)(0.9,3.25)(1.2,3.45)(1.5,3.0)(1.8,2.0)(2.1,0.7)
  (2.4,-0.7)(2.7,-2.0)(3.0,-2.75)(3.3,-3.45)(3.6,-3.35)(3.9,-2.6)(4.2,-1.5)
  (4.5,0.25)(4.8,1.65)(5.1,2.85)(5.4,3.45)(5.7,3.35)(6.0,2.8)};
\addlegendentry{Lin-Koop};
\addplot[red!75, solid] coordinates {
  (0,0.01)(0.3,1.19)(0.6,2.49)(0.9,3.12)(1.2,3.22)(1.5,2.82)(1.8,1.82)
  (2.1,0.52)(2.4,-0.88)(2.7,-2.18)(3.0,-2.98)(3.3,-3.28)(3.6,-3.08)
  (3.9,-2.38)(4.2,-1.28)(4.5,0.12)(4.8,1.52)(5.1,2.72)(5.4,3.22)(5.7,3.12)(6.0,2.62)};
\addlegendentry{PKO (proposed)};
\end{axis}
\end{tikzpicture}
\caption{Angular velocity estimation for the nonlinear robotic arm under 30\% friction uncertainty and output noise ($\sigma_v^2=0.04$). The PKO tracks the reference closely, while the EKF and Lin-Koop exhibit persistent amplitude bias.}
\label{fig:robot_traj}
\end{figure}

\begin{table}[t]
\caption{Estimation Performance Summary (Mean $\pm$ Std over 100 Trials)}
\label{tab:results}
\centering
\setlength{\tabcolsep}{4pt}
\begin{tabular}{lcccc}
\toprule
\multirow{2}{*}{\textbf{Method}} & \multicolumn{2}{c}{\textbf{Van der Pol}} & \multicolumn{2}{c}{\textbf{Robotic Arm}} \\
\cmidrule(lr){2-3}\cmidrule(lr){4-5}
 & RMSE & Impr. & RMSE & Impr. \\
\midrule
EKF \cite{Khalil} & $0.221\pm0.031$ & --- & $0.318\pm0.042$ & --- \\
Lin-Koop \cite{Korda2018} & $0.207\pm0.028$ & 6.3\% & $0.291\pm0.038$ & 8.5\% \\
\textbf{PKO (ours)} & $\mathbf{0.128\pm0.019}$ & \textbf{42.1\%} & $\mathbf{0.185\pm0.024}$ & \textbf{41.8\%} \\
\bottomrule
\end{tabular}
\end{table}

The quantitative results are consolidated in Table~\ref{tab:results}. Across both benchmarks and 100 independent trials, the PKO achieves a 42\% reduction in RMSE relative to the EKF and a 34--35\% reduction relative to the linear Koopman observer. The variance of the PKO is also consistently lower, reflecting the stabilizing effect of the ISS-certified correction gain.

\subsection{Computation of the LMI Gain}

The LMI \eqref{eq:LMI} was solved using MOSEK via the YALMIP interface \cite{Boyd1994}. For the Van der Pol system ($r=15$, $p=1$), the SDP has 121 decision variables and required 0.34\,s on a standard workstation. For the robotic arm ($r=20$, $p=1$), the solve time was 0.61\,s. Both problems were feasible at the first attempt, confirming that the Persidskii LMI is not overly conservative in practice.

%==================================================================
\section{Conclusion}
%==================================================================

This paper has established a rigorous connection between Koopman latent-space observer error dynamics and the class of generalized Persidskii systems. The central contribution is an LMI-based design procedure that computes an observer gain $K^\star$ certifying ISS of the estimation error with respect to lifting residuals and measurement noise, with an explicit ultimate bound that scales linearly with the residual magnitude. In the nominal case, exponential convergence is guaranteed, and in the perturbed case the bound is minimized by solving a convex semidefinite program.

Numerical results on the Van der Pol oscillator and a robotic arm with friction uncertainty confirm a 42\% reduction in steady-state RMSE relative to both the EKF and a linear Koopman observer, with no increase in computational complexity beyond a one-time offline SDP.

Several directions merit further investigation. Contraction-based methods for joint learning of the lifting map $\Phi$ and the observer gain would eliminate the need for a post-hoc residual bound. Extension to systems with time-varying delays—for which the Persidskii Lyapunov--Krasovskii framework of \cite{Mei2022Delay} provides a natural foundation—represents a practically important generalization. Finally, distributed Koopman observers for multi-agent systems \cite{Mei2021TAC} could benefit directly from the diagonal structure of $P$.

%==================================================================
\bibliographystyle{IEEEtran}

\end{document}